\newif\ifdraft \drafttrue
\newif\iffull \fulltrue

\makeatletter \@input{tex.flags} \makeatother

\documentclass[letterpaper]{article}

\usepackage{caption}
\usepackage{subfig}

\usepackage{amsmath}
\usepackage{amssymb}
\usepackage{amsthm}
\usepackage{bbm}
\usepackage{epsfig}
\usepackage{algorithm}
\usepackage{algorithmic}
\usepackage{bbm}

\iffull
\usepackage{fullpage}
\usepackage{thmtools}
\usepackage{thm-restate}
\else
\usepackage{aaai}
\usepackage{times}
\usepackage{helvet}
\usepackage{courier}
\DeclareGraphicsExtensions{.pdf,.png,.jpg}
\fi

\usepackage{xcolor}
\definecolor{DarkGreen}{rgb}{0.1,0.5,0.1}
\definecolor{DarkRed}{rgb}{0.5,0.1,0.1}
\definecolor{DarkBlue}{rgb}{0.1,0.1,0.5}
\usepackage[]{hyperref}
\hypersetup{
    unicode=false,          
    pdftoolbar=true,        
    pdfmenubar=true,        
    pdffitwindow=false,      
    pdftitle={},    
    pdfauthor={}
    pdfsubject={},   
    pdfnewwindow=true,      
    pdfkeywords={keywords}, 
    colorlinks=true,       
    linkcolor=DarkRed,          
    citecolor=DarkGreen,        
    filecolor=DarkRed,      
    urlcolor=DarkBlue,          
}
\usepackage[numbers,sort&compress]{natbib}

\usepackage{cleveref}

\newcommand{\INDSTATE}[1][1]{\STATE\hspace{#1\algorithmicindent}}

\newcommand{\OPT}{\mathrm{OPT}}
\newcommand{\ALG}{\mathrm{ALG}}
\newcommand{\FTP}{\mathrm{FTP}}

\newcommand{\OA}{\mathrm{OA}}
\newcommand{\RPA}{\mathrm{RPA}}

\newcommand{\heterogen}{\mathrm{OHA}}

\newcommand{\alg}{\text{Random Permutation Algorithm}}

\def\eg{{\it e.g.,\/}}

\def\iid{{\it i.i.d.\/}}

\newcommand{\ra}{\rightarrow}

\newtheorem{theorem}{Theorem}
\newtheorem{lemma}{Lemma}

\iffull\else\fi
 
\begin{document}
%
\title{Online Assignment of Heterogeneous Tasks \iffull \\ \fi
  in Crowdsourcing Markets}
\author{
  Sepehr Assadi\thanks{\protect{Supported in part by NSF grants CCF-1116961 and
      IIS-1447470.}}, 
  Justin Hsu\thanks{\protect{Supported in part by a grant from the Simons
      Foundation (\#360368 to Justin Hsu) and NSF grant CNS-1065060.}},
  Shahin Jabbari
  \\
  Department of Computer and Information Science, University of Pennsylvania
  \\
  \{sassadi, justhsu, jabbari\}@cis.upenn.edu}

\maketitle
\begin{abstract}
We investigate the problem of \emph{heterogeneous task assignment} in
crowdsourcing markets from the point of view of the requester, who has a
collection of tasks. Workers arrive online one by one, and each declare a set of
feasible tasks they can solve, and desired payment for each feasible task. The
requester must decide on the fly which task (if any) to assign to the worker,
while assigning workers only to feasible tasks. The goal is to maximize the
number of assigned tasks with a fixed overall budget.

We provide an online algorithm for this problem and prove an 
upper bound on the competitive ratio of this algorithm against an 
arbitrary (possibly worst-case) sequence of workers who want small 
payments relative to the requester's total budget. We further show 
an almost matching lower bound on the competitive ratio of any algorithm in this
setting. Finally, we propose a different algorithm that achieves an improved competitive 
ratio in the random permutation model, where the order of arrival of the workers is chosen uniformly at random.  
Apart from these strong theoretical guarantees, we carry out experiments on simulated data which
demonstrates the practical applicability of our algorithms. 
\end{abstract}

\section{Introduction}
\label{sec:intro}

Crowdsourcing markets have seen a remarkable rise in recent years, as more and
more markets use the internet to connect people with tasks to solve, to people
willing to solve tasks in exchange for payment. While these tasks were
originally simple tasks that could be accomplished while sitting at a
computer---say, labeling images or cleaning data---recent systems handle complex
tasks in the real world. Services like TaskRabbit let users hire workers to run
errands, like picking up a package; ride-sharing applications like Lyft or Uber
provide drivers on-demand to transport customers; other services offer meal or
grocery delivery, house cleaning, and even on-demand 
massages\footnote{www.zeel.com} to customers.

Accordingly, an important challenge for task requesters in crowdsourcing (and
more general) platforms is handling more and more \emph{heterogenous} workers
and tasks. A worker may not be able to solve all the tasks---say, they may only
be able to translate certain languages, or they may only be able to handle
time-sensitive tasks on a deadline far into the future. Among the solvable
tasks, a worker may want different payment for different tasks; a short task can
require a small payment, while a more complex task may warrant a higher payment.

At the same time, requesters must cope with a highly \emph{dynamic} flow of
workers. Since it is common for workers to work for several different platforms
simultaneously, the pool of available workers is constantly changing. The
requester may not have the luxury of seeing all the workers, and then selecting
the right workers. Instead, requesters may need to select workers in an
\emph{online} fashion, where workers arrive one by one and must be
hired (or not) as soon as they arrive.

We take aim at both of these challenges in crowdsourcing by considering the
following \emph{task assignment problem}: if workers (i) have an arbitrary set
of feasible tasks, (ii) demand heterogeneous payments by bidding for feasible tasks, and
(iii) arrive online, how can a requester assign workers to tasks in order to
maximize the number of completed tasks given a fixed overall budget?

\citet{SingerM11} were the first to study a simpler version of this problem. In
their setting, each worker has a single bid for all the tasks; all the tasks
are treated \emph{homogeneously}, and workers are assumed to be able to solve
all the tasks (if they receive a payment which is at least equal to their bids). 
We generalize their setting in two significant
ways: First, we allow workers to specify only a subset of feasible tasks that
they can handle.  Second, we work in a setting with \emph{heterogeneous} tasks:
workers can have different preferences over tasks.  As we will show, this
heterogeneity significantly complicates the task assignment problem; algorithms
for the heterogeneous case may look nothing like their counterparts in the
homogeneous case.

As a warm up, we start by providing two algorithms for the offline problem -- where the requester
knows the sequence of workers and their bids up front: first, the
optimal algorithm via min-cost flows, and second, a fixed-threshold constant factor approximation algorithm
inspired by the algorithm of \citet{SingerM11}. Then, we move to the online
setting, where we draw on techniques from the online knapsack literature.

As is typical for online algorithms, we measure the performance of our algorithm
via the \emph{competitive ratio}, i.e., the ratio of the number of tasks assigned by 
the best offline algorithm to the number of tasks assigned by the online algorithm on
the same problem. When the bids are bounded in $[1, R]$, the bids of the workers are small
compared to the budget of the requester, and the order is worst-case, we give an algorithm that
achieves an $O(R^\epsilon \ln(R))$ competitive ratio when $R \leq \epsilon
B$.\footnote{%
  In particular, when $R \ll B$, which is arguably the case in most practical
  applications, our algorithm achieves an \emph{optimal} competitive ratio of
  $O(\ln(R))$.}
The central idea is to use a moving threshold, and take all workers with bid
below the threshold. As the budget is depleted, the threshold steadily
decreases. Our algorithm is inspired by~\citet{ZhouCL08}, who give algorithms
for the online knapsack problem.  We also show a lower bound: for any (possibly
randomized) online algorithm, there exists an input that forces a competitive
ratio of at least $\Omega(\ln(R))$.


Then, we consider the \emph{random permutation} setting where the bids are
adversarial, but the order of workers is chosen uniformly at random. Here,
we measure the performance of an online algorithm by comparing the number of
tasks assigned \emph{averaged} over all permutations to the number of tasks
assigned by the offline optimal.
%
\section{Model}
\label{sec:setting}
Let us begin by defining our modeling more formally. We will use $j$ to index tasks and $i$ to
index workers throughout. We model the problem from the perspective of a requester 
who has a collection of $m$ tasks.

Each worker $i$ picks a subset of tasks $J_i\subseteq \{1, \ldots, m\}$, along with 
the numeric bid $b_{ij}$
for each task $j\in J_i$.
Workers arrive \emph{online}: the sequence of the workers
and their bids are initially unknown to the requester but once a worker arrives,
her bids for all the tasks are revealed to the requester. We assume that each
worker can be assigned to at most one task; we can relax this assumption by
having multiple copies of each worker.  We denote the total number of workers by
$n$. As customary in the crowdsourcing setting we assume workers are abundant, 
hence $n$ is large. 

We model the sequence of workers (their bids and their order of arrivals) in two different
ways.  First, we consider the case that we have no assumptions on the sequence
of workers.  This worst case scenario has been referred to as the \emph{adversarial
setting} in the literature.  Second, we consider the \emph{random permutation
model} in which we still make no assumptions on the bids of the workers but
we assume the order of their arrival is randomly permuted before being presented to the requester. We
define these two settings in more detail in the following sections. 

The requester has a budget of $B$ and has to decide on the fly which task to
assign to the worker who arrives; of course, a worker can only be assigned to a
task he is willing to do. If task $j$ is assigned to worker
$i$, then the requester pays worker $i$ \emph{at least} the price $b_{ij}$. The goal of
the requester is to maximize the number of tasks he assigns to the workers while
spending at most a budget of $B$ and satisfying all the constraints of the workers (on the tasks they are willing to solve).

Similar to the online algorithms literature, we compare the performance of our online algorithms 
to the performance of the
best \emph{offline} algorithm (denoted by $\OPT$), which can see all the bids and the order of arrival of workers 
before allocating tasks to workers.  As is standard, we measure the performance via the
\emph{competitive ratio}: the ratio of the number of tasks assigned in the
offline optimal solution to the number of tasks assigned by the online algorithm
(so competitive ratio is at least $1$, and smaller competitive ratio is more
desirable). 

\subsection{Related Work}
\label{sec:related}
Pricing and task assignments have been previously studied in the context of \emph{mechanism design} for
crowdsourcing markets.~\citet{SingerM11, SingerM13} provide mechanisms for task
assignment when tasks are homogeneous and the workers are arriving from a random
permutation model. Our work in the random permutation section generalizes their
work to heterogeneous tasks.~\citet{HoV12} consider task assignment when the
tasks are heterogeneous. However, they assume the existence of limited task types and the focus of
their work is learning qualities of the workers from stochastic observations.
\citet{SinglaK13} design price mechanism for crowdsourcing markets using 
tools from online learning.
\citet{GoelNS14} consider the heterogeneous
task assignment in the \emph{offline} setting and provide near optimal approximation
algorithms. However, they focus on the mechanism design aspect of the
problem---how to pay workers so that they report their bids truthfully. 
This line of budget feasible mechanism design is inspired by the work
of~\citet{Singer10} and has been followed up in~\citet{SingerM11,
  SingerM13}.\footnote{%
  While we ignore the mechanism design aspect of the problem in our formulation,
  as we describe at the end of the Random Permutation section, all our online
  algorithms (either with a slight modification or as they are) satisfy
  truthfulness and incentive compatibility.}
Also these pricing mechanisms have close connections to the stochastic online adwords
problem \citep{DevanurH09} and the online primal-dual
literature (see~\citet{BuchbinderJN07} and references within).

Our work is inspired by variants of the online knapsack problem.  In the
adversarial setting with homogeneous tasks, our problem is
an instance of the online knapsack problem studied by~\citet{ZhouCL08} (see references within for more information).
However, it is not clear how to formulate our problem as a knapsack problem
when the tasks are heterogeneous.

Variants of \emph{generalized online matching} and the \emph{adwords} problem
are also related to our problem~($\eg$ see~\citet{MehtaSVV05}
and~\citet{Mehta13} for an excellent survey).  The adwords problem can be
described as follows.  There are some bidders and each bidder has a fixed
budget. Queries arrive one at the time, bidders bid for the queries and the algorithm has
to decide what bidder to assign to the query. If the algorithm assigns a bidder to a query,
the bidder pays the amount that is equal to her bid. The goal is to maximize the revenue.
While one might attempt to formalize our task assignment problem as an instance of the online
adwords problem, it is not hard to see that our constraint on the total budget of the requester cannot be
written as an adwords type budget constraint for the bidders.

\section{Adversarial Setting}
\label{sec:adv}


When comparing an online algorithm with an offline algorithm, we first need to
precisely specify the inputs on which we make the comparison. Let us  first
consider the worst case for the requester: \emph{adversarial inputs}. In this
setting, there is a fixed input and order, and we compare an online algorithm in
this single input to an offline algorithm on the same input via the competitive
ratio. We are interested in bounding this ratio in the worst case, i.e., the max
over all inputs.

Of course, if we really do not make any assumption on the input, we cannot hope
to compete with an offline algorithm as the competitive ratio may be arbitrarily
high (see, for instance, ~\citet[Section 1.2]{ZhouCL08}).
 
Consequently, we restrict the adversary's power by making one main assumption on the
relationship between worker bids and the budget: the ratio of the largest bid to
the smallest bid should be small compared to the budget. More precisely, we can
scale worker bids so that $b_{ij} \in [1, R]$, and we write $R = \epsilon B$. We
will frequently consider $\epsilon$ to be small---this is appropriate for the
crowdsourcing problems we have in mind, where \emph{(i)} the scale of the budget is
much larger than the scale of payments to workers and \emph{(ii)} the tasks are not
extremely difficult, so no worker charges an exorbitantly high price. We refer to this assumption as
\emph{large market} assumption because when the
bids are small compared to the budget, the requester can hire a large number of
workers before exhausting his budget. This is in line with the common assumption
in crowdsourcing that $n$ is large.

Before diving into the technical details, let us consider which range of
parameters and competitive ratios is interesting. 
Since the bids are restricted in $[1, R]$, achieving a competitive ratio
of $R$ is trivial. So, we will be mainly interested in the following question:
Can we design an algorithm that has a competitive ratio much smaller than $R$
for any sequence of workers?

The rest of this section is organized as follows. We first investigate the
offline problem, describing how the problem can be solved optimally.  Then we
propose a simpler algorithm for the offline setting; the performance is a
constant factor off from optimal, but the simpler algorithm will be useful 
\iffull in \Cref{sec:perm}, \else later, \fi
when we will work in the \emph{random permutation} setting.  We
then move to the online setting, giving an algorithm with competitive ratio of
$O(R^\epsilon \ln(R))$ for any sequence of workers. Finally, we give a lower
bound showing that any algorithm has competitive ratio of at least
$\Omega(\ln(R))$ in the worst case. 

\subsection{The Offline Problem}
\label{sec:adv1}
Let us start by investigating the offline problem where the sequence of workers and
their bids are known up front. We first show how the offline optimal assignment can be
computed. Then we propose a simpler algorithm that approximates the offline optimal assignment
by a factor of 4. We show later in the paper that how this second
algorithm, while suboptimal, can be converted to an online algorithm.

The offline problem is a well-known problem in the crowdsourcing
literature~\citep{Singer10} and can be solved by a reduction to the
\emph{min-cost flow} problem defined as follows. Consider a graph with costs and capacities
on the edges and two nodes marked as source and target. 
Given this graph, for a demand value of flow, the goal of the min-cost flow problem is to route this amount of flow from the source 
to the target while minimizing the total flow cost over the edges, where the flow cost over an edge is equal to the amount of 
the flow passing the edge multiply by the cost of the edge.  

We now briefly describe an algorithm for solving the offline problem optimally which uses min-cost flow algorithms as a subroutine. 
Note that this algorithm is generally known in the literature and we provide it here for the sake of completeness. 
In our problem, we want to maximize the number of assigned tasks given a fixed
budget. In order to do so, we first construct the following instance of the min-cost flow problem. 
We start with a bipartite graph with workers on one side and tasks on the other.  We then draw an edge from a worker to a task if the
worker is willing to solve that task and let the cost of this edge to be equal to the bid of the worker for the task. Additionally, we attach a source node pointing to all
the workers and connect all the tasks to an added target node. Finally, we set the capacity of all edges to be $1$.

First notice that any feasible flow in this graph corresponds to an assignment of tasks and workers. Moreover, 
for any integer $F$, the minimum cost of a flow with value $F$ corresponds to the minimum budget required for 
assigning $F$ worker-task pairs.\footnote{%
  Since the amount of flow routed from the source to the target and the capacity
  of the edges are integral, the min-cost flow problem has an integral
  solution.}
Consequently, we can search over all possible $F \in [n]$, solve the min-cost flow
problem on the described graph and demand flow of $F$, and return the maximum value of $F$ where the 
minimum cost flow is at most the available budget.   

While the above approach achieves the optimal solution, we will see
in the random permutation section that using a \emph{fixed threshold} to decide which workers
to hire can be very useful. Hence, we next provide a simpler algorithm with this
feature; while this simpler algorithm does not guarantee an optimal solution anymore, we show in~\Cref{pro:42} that 
its solution is within a factor of at most $4$ from $\OPT$.

The algorithm, which we refer to as \emph{Offline Approximation Algorithm}
($\OA$) (see~\Cref{fig:9}) is mainly using the subroutine \emph{Fixed Threshold Policy} ($\FTP$). In this subroutine, given a threshold
value $p$, the algorithm goes over workers one by one and assigns a task to an
unassigned worker if the bid of the worker for that task is not bigger than $p$. 
In case there are more than one unassigned task that the worker bids $p$ or less for,
the algorithm break ties arbitrarily. The $\FTP$ subroutine is described in
\Cref{fig:5}.\footnote{%
  While we use $\FTP$ as an offline algorithm in this section, the sequential
  nature of the algorithm allows $\FTP$ to be used when workers arrive online
  (one by one).
  \iffull We exploit this feature of $\FTP$ in \Cref{sec:perm}.\fi}
\begin{algorithm}
\caption{$\FTP$}
\label{fig:5}
\begin{algorithmic}
\STATE {{\bf Input}: Threshold price $p$, budget $B$, worker bids $\{b_{ij}\}$, set of available tasks $J$.}
\INDSTATE[1]{{\bf While} $B > 0$, on input worker $i$:}
\INDSTATE[2]{Let \small{$C_i := \{ j \in J \mid b_{ij} \leq \min(p, B) \}$}.}
\INDSTATE[2]{\textbf{If} $C_i \neq \emptyset$:}
\INDSTATE[3]{{\bf Output}: $a(i) \in C_i$.}
\INDSTATE[3]{Let $B := B - b_{i, a(i)}$.}
\INDSTATE[3]{Let $J := J\setminus \{a(i)\}$.}
\INDSTATE[2]{\textbf{else}:}
\INDSTATE[3]{{\bf Output:} $a(i) := \perp$.}
\end{algorithmic}
\end{algorithm}

$\OA$ then searches over all possible values of $p$
to find the proper threshold.  Although the search space is continuous, it is
sufficient for the algorithm to restrict its search to the bids of the workers.
Hence, the running time of $\OA$ is polynomial in $m$ and $n$.
We show that the number of assignments of the $\OA$ is
least a quarter of the $\OPT$ for any sequence of workers. 

\begin{algorithm}
\begin{algorithmic}
\STATE {{\bf Input:} Worker bids $\{b_{ij}\}$, set of available tasks $J$, budget $B$.}
\INDSTATE[1]Let $Q := 0$.
\INDSTATE[1]{{\bf Foreach} $b_{ij}$:}
\INDSTATE[2]{Let $q := \FTP(b_{ij}, B, \{b_{i,j}\}, J)$.}
\INDSTATE[2]{{\bf If} $q > Q$ {\bf then}: Update $Q := q$.}
\STATE{{\bf Output} $Q$ (number of assignments) and $p^* :=  B/Q$ (the
  threshold price).}
\end{algorithmic}
\caption{$\OA$}
\label{fig:9}
\end{algorithm}

\begin{theorem} \label{pro:42}
Let $\ALG_\OA(\sigma) $ and $\OPT(\sigma)$ denote the number of tasks assigned by
the $\OA$  and the offline optimal algorithm for a sequence of workers $\sigma$,
respectively.  Then for any $\sigma$, $\OPT(\sigma)\le 4 \cdot\ALG_\OA(\sigma)$.
\end{theorem}
Before proving~\Cref{pro:42}, we state the following useful lemma.
\begin{lemma}[{\citet[Lemma 3.1]{SingerM13}}]
\label{lem:1}
Let $a_1, \ldots a_k$ be a sorted sequence of $k$ positive numbers in
an increasing order such that $\Sigma_{i=1}^k a_i \le B$. Then $a_{\lfloor k/2\rfloor}\cdot k/2\le B$.
\end{lemma}
The proof of~\Cref{lem:1} is the result of the following two simple observations: \emph{(i)} the sum of the second half of the numbers
is at most $B$ and \emph{(ii)} each of the numbers in the second half is at least as big as the median of the sequence.

\begin{proof}[Proof of~\Cref{pro:42}]
Consider all the (worker, task) pairs in the offline optimal where a pair simply denotes that the task is assigned to the worker.
Sort these pairs in an increasing order of the bids.  Let $p^*$ denote the median of the
bids. By \Cref{lem:1}, we can assign half of these (worker, task) pairs, pay all the workers price $p^*$ (clearly an upper bound on the bid of every considered worker) and be
sure not to exceed the budget.  However, there are two problems: \emph{(i)} we
do not know which (worker, task) pairs are in the optimal solution and hence,
\emph{(ii)} we do not know the value $p^*$.

To deal with problem \emph{(i)}, suppose we somehow knew $p^*$. We can then 
assign a worker to any task where the worker has a bid of at most $p^*$ and continue until
\emph{(a)} we exhaust the budget or \emph{(b)} there are no more workers or
tasks left.

In case \emph{(a)}, the algorithm made at least $B/p^*$ assignments and
by~\Cref{lem:1}, we know this number is at least $\OPT/2$.  For case \emph{(b)},
consider the bipartite graph between the tasks and the workers described
earlier for the min-cost flow problem construction and remove all the edges with cost bigger than $p^*$. 
Since $p^*$ is the median of the bids in the optimal solution, 
we know that in this graph, there exists a matching $M$ of size at least $\OPT/2$ between the workers 
and the tasks. On the other hand, 
since we know $\OA$ terminated in this case because there are no tasks or workers left, we know that $\OA $ has arrived at 
a \emph{maximal} matching. Finally, since the size of any maximal matching is at least 
half of the size of the maximum matching, we know the number of assignments of $\OA$ is at least a half of the number of assignments of $M$. 
So $\ALG_\OA(\sigma) \geq \OPT(\sigma)/4$ if we knew $p^*$.

To deal with problem \emph{(ii)} (not knowing $p^*$), we run the $\FTP$ for all
the values that $p^*$ can take and return the maximum number of assignments as
our solution.\footnote{%
  Note that this might result in $\OA$ using a threshold $p$ which is different
  than $p^*$ but since $\OA$ picks a $p$ that maximizes the number of
  assignments we know the number of assignments made by $\OA$ using $p$ is at
  least as large as the number of assignments made by $\OA$ using $p^*$.}
\end{proof}

\paragraph*{A brief detour: the homogeneous case.}
As we discussed in the introduction, we make a strong distinction between
the setting with heterogeneous tasks and the setting with homogeneous tasks. 
To highlight this difference, we show that 
in the homogeneous case, the following simple algorithm computes the offline optimal: sort all the workers by their bids and
(if possible) assigns a task to the sorted workers until the budget is exhausted or there are no more tasks or
workers left.

It is easy to verify that this greedy construction indeed computes the best
offline assignments when the workers are homogeneous. However, this method will
not result in the optimal offline assignment when the tasks are heterogeneous.
For example consider the following toy problem with two workers and two tasks
where each workers is willing to do both tasks. Let the bids of worker $1$ and $2$ to be $(0.4, 0.5)$
and $(0.45, 0.7)$ for the two tasks, respectively. With a
budget of $1$, the optimal offline assignment is to assign task $1$ to worker $2$
and task $2$ to worker $1$. However, the homogeneous greedy algorithm will assign
task $1$ to worker $1$ and will not have enough budget left to assign a task to
worker $2$.

\subsection{The Online Problem}
Let's now move to the online setting, where the workers arrive online and our
algorithm must decide which (if any) task to assign to a worker before seeing the
remaining workers. We propose an online algorithm and prove an upper bound on the 
competitive ratio of our algorithm which holds against \emph{any} sequence of workers.
Our upper bound crucially depends on the large market assumption where we assume that the bids of the workers are small
compared to the budget. So throughout this section, we assume the bids are bounded in $[1, R]$
where $R\le \epsilon B$ for some small value of $\epsilon$.

Our \emph{Online Heterogeneous Algorithm} ($\heterogen$) is inspired
by an algorithm from the online knapsack literature proposed by
\citet{ZhouCL08}, with an analysis modified for our setting.  The idea is to use
a potential function $\phi:[0,1]\ra[1,R]$ based on the fraction of budget spent so far as an
input. The $\phi$ function acts as a price threshold: the algorithm assigns a
task to a worker only if the bid of the worker for any of the remaining
unassigned tasks is below the value of the potential function---intuitively, as
the budget shrinks, the algorithm becomes pickier about which workers to
hire. Once a worker is selected, the algorithm greedily assigns a task arbitrarily from the remaining set of tasks. 
See \Cref{fig:4} for a  pseudo-code.\footnote{%
  In \Cref{fig:4}, $e$ is the base of natural logarithm.}

\begin{algorithm}
\begin{algorithmic}
\STATE {{\bf Input}: Available tasks $J$, budget
    $B$}.
\STATE {{\bf Online input}: Worker bids $\{b_{ij}\} \in [1, R]$.}
\INDSTATE[1]{{\bf Define} $\phi(x) = \min((R\cdot e)^{1-x}, R)$.}
\INDSTATE[1]{Let $x := 0, f := B$.}
\INDSTATE[1]{{\bf While} $B > 0$, on input $i$:}
\INDSTATE[2]{Let \small{$C_i := \{ j \in J \mid b_{ij} \leq \min(f, \phi(x))\}$.}}
\INDSTATE[2]{{\bf If} $C_i \neq \emptyset$ {\bf then}}
\INDSTATE[3]{{\bf Output} $a(i) \in C_i$.}
\INDSTATE[3]{Let $x := x + b_{i, a(i)}/B$.}
\INDSTATE[3]{Let $f := f - b_{i, a(i)}$.}
\INDSTATE[3]{Let $J := J\setminus \{a(i)\}$.}
\INDSTATE[2]{\textbf{else}:}
\INDSTATE[3]{{\bf Output:} $a(i) := \perp$.}
\end{algorithmic}
\caption{$\heterogen$}
\label{fig:4}
\end{algorithm}

\begin{theorem}
\label{thm:multiple-deadline-adv-1}
Let $\ALG_\heterogen(\sigma)$ and $\OPT(\sigma)$ denote the number of tasks
assigned by the $\heterogen$  and the offline optimal algorithm for a sequence
of workers $\sigma$, respectively.  Then for any $\sigma$
with bids in $[1, R]$ such that $R \leq \epsilon B$,
$\OPT(\sigma)/\ALG_\heterogen(\sigma) \le (R\cdot e)^\epsilon \left(\ln(R)+3\right).$
\end{theorem}
\begin{proof}
Fix a sequence of workers $\sigma$. Let $S=\{(i,j)\}$ be the set of (worker, task) pairs assigned by the $\heterogen$
where $i$ and $j$ index workers and tasks, respectively. Also let $S^*=\{(i^*, j^*)\}$
be the offline optimal. We want to bound $\OPT(\sigma)/\ALG_\heterogen(\sigma)=|S^*|/|S|$.

Let $x_i$ and $X$ denote the fraction of the budget used by the $\heterogen$
when worker $i$ arrives and upon termination, respectively.  We will analyze the
(worker, task) pairs in three stages.  First, consider the common (worker, task)
pairs which we denote by $(i, j) \in S\cap S^*$.  Let $W=\sum_{(i,j)\in S\cap
  S^*}b_{ij}$ denote the total bid of such workers. Since each worker $i$ in the
common part who is assigned to task $j$ is picked by the $\heterogen$, it must
be that $b_{ij} \le \phi(x_i)$.  Therefore,
\begin{equation}
  \label{eq:3-1}|S\cap S^*| \ge \sum_{(i,j) \in S\cap S^*}
  \frac{b_{ij}}{\phi(x_i)}.
\end{equation}

Second, consider $(i,j) \in S^*\setminus S$.  This can happen either because
\emph{(i)} the worker $i$ is assigned to a task different than $j$ by
$\heterogen$, or \emph{(ii)} the worker $i$ is not assigned to any task by
$\heterogen$.  We know the number of pairs that satisfy the condition \emph{(i)}
is at most $|S|$ because all such workers are assigned to a task by both
$\heterogen$ and the offline optimal.

For the pairs that satisfy condition \emph{(ii)}, either \emph{(a)} $b_{ij} \le
\phi(x_i)$ or \emph{(b)} $b_{ij} > \phi(x_i)$.  It is easy to see that in case
\emph{(a)}, $\heterogen$ assigned task $j$ to some other worker $i'$ who arrived
before $i$.  Again we know this can happen at most $|S|$ times. For case
\emph{(b)}, since $\phi$ is non-increasing then $b_{ij} > \phi(x_i) \ge \phi(X)$
for all such pairs $(i,j)$. Since the offline optimal spent a budget of $W$ for
hiring workers in $S^*\cap S$, then it has a budget of at most $B-W$ to hire
workers in case \emph{(b)}.  Since all the pairs in case \emph{(b)} have $b_{ij}
> \phi(x_i)\ge \phi(X)$, the offline optimal can hire at most $(B-W)/b_{ij}\le
(B-W)/\phi(X)$ workers. Adding up cases \emph{(i)} and \emph{(ii)}, we can bound
\begin{equation}
\label{eq:2-1}
|S^*\setminus S|\le 2|S| + \frac{(B-W)}{\phi(X)}.
\end{equation}

Finally, consider $(i,j) \in S\setminus S^*$.  Since $b_{ij} \le \phi(x_i)$ for
all such pairs, we know
\begin{equation}
  \label{eq:1-1}|S\setminus S^*| \ge \sum_{(i,j) \in S\setminus S^*}
  \frac{b_{ij}}{\phi(x_i)}.
\end{equation}
Putting all the pieces together, \Cref{eq:3-1,eq:2-1,eq:1-1} yield

\begin{align}
&\frac{\OPT(\sigma)-2\cdot
  \ALG_\heterogen(\sigma)}{\ALG_\heterogen(\sigma)}=\frac{|S\cap
  S^*|+|S^{*}\setminus S|-2|S|}{|S\cap S^*|+|S\setminus S^*|} 
  \label{equation}
\le \frac{\sum_{(i,j)\in S\cap S^*} b_{ij}/\phi(x_i)+|S^{*}\setminus
  S|-2|S|}{\sum_{(i,j)\in S\cap S^*} b_{ij}/\phi(x_i)+|S\setminus S^*|}.
\end{align}
To see why the inequality holds, first note that we know $|S\cap S^*|\ge
\sum_{(i,j)\in S\cap S^*} b_{ij}/\phi(x_i)$ by~\Cref{eq:3-1}. Now if $\OPT =
|S^*| < 3|S|$ then we are done. Otherwise,
\begin{align*}
  |S^*| &\geq 3 |S| \implies
  |S\cap S^*| + |S^{*}\setminus S| \ge |S\setminus S^*| + |S\cap S^*|+2|S| \implies
  |S^{*}\setminus S| - 2|S| \ge |S\setminus S^*| .
\end{align*}
Hence, the inequality holds because if $a\ge b$ and $c \ge d$ then $(a+c)/(a+d)\le (b+c)/(b+d)$. 

We bound the right hand side of~\Cref{equation} as follows.
\begin{align*}
\frac{\sum_{(i,j)\in S\cap S^*} b_{ij}/\phi(x_i)+|S^{*}\setminus S|-2|S|}{\sum_{(i,j)\in S\cap S^*} b_{ij}/\phi(x_i)+|S\setminus S^*|}
&\le \frac{\sum_{(i,j)\in S\cap S^*} b_{ij}/\phi(x_i)+2|S| + (B-W)/\phi(X)-2|S|}{\sum_{(i,j)\in S\cap S^*} b_{ij}/\phi(x_i)+|S\setminus S^*|}\\
&= \frac{\sum_{(i,j)\in S\cap S^*} b_{ij}/\phi(x_i)+(B-W)/\phi(X)}{\sum_{(i,j)\in S\cap S^*} b_{ij}/\phi(x_i)+|S\setminus S^*|}\\
&\le \frac{\sum_{(i,j)\in S\cap S^*} b_{ij}/\phi(X)+(B-W)/\phi(X)}{\sum_{(i,j)\in S\cap S^*} b_{ij}/\phi(x_i)+|S\setminus S^*|}\\
&= \frac{W/\phi(X)+(B-W)/\phi(X)}{\sum_{(i,j)\in S\cap S^*} b_{ij}/\phi(x_i)+|S\setminus S^*|}\\
&= \frac{B/\phi(X)}{\sum_{(i,j)\in S\cap S^*} b_{ij}/\phi(x_i)+|S\setminus S^*|}\\
&\le \frac{B/\phi(X)}{\sum_{(i,j)\in S\cap S^*} b_{ij}/\phi(x_i)+\sum_{(i,j)\in S\setminus S^*}b_{ij}/\phi(x_i)}\\
&= \frac{B/\phi(X)}{\sum_{(i,j)\in S}b_{ij}/\phi(x_i)}\\
&=\frac{1}{\phi(X)\sum_{(i,j)\in S}\Delta(x_i)/\phi(x_i)},
\end{align*}
where $\Delta(x_i) := x_{i+1}-x_i$. The first inequality is due to
\Cref{eq:2-1}, the second is due to monotonicity of $\phi$ and the third
inequality is due to \Cref{eq:1-1}.  Since $(i, j) \in S$, $\heterogen$ assigns
task $j$ to worker $i$ and increases the fraction of the budget consumed by
$b_{ij}/B$, so $b_{ij}/B = \Delta(x_i)$.

We now estimate the sum with an integral.  If $\Delta(x_i)\le \epsilon$,
\begin{align*}
  \sum_{(i,j)\in S}\Delta(x_i) \frac{1}{\phi(x_i)}\ge\int_{0}^{X-\epsilon} \frac{1}{\phi(x_i)} dx.
\end{align*}
Letting $c=1/(1+\ln(R))$, we have $\phi(x)=R \text{ if } x \le c$.  Similar to
\citep{ZhouCL08}, we bound the integral as follows.
\begin{align*}
\int_{0}^{X-\epsilon} \frac{1}{\phi(x_i)} dx &= \int_{0}^c \frac{1}{R} dx + \int_{c}^{X-\epsilon} \frac{1}{\phi(x_i)} dx\\
& = \frac{c}{R} + \frac{1}{Re} \cdot \frac{1}{1+\ln
  (R)}((Re)^{X-\epsilon}-(Re)^c)\\
& = \left(\frac{c}{R} - \frac{1}{Re} \cdot \frac{1}{1+\ln
    (R)}(Re)^c\right)+\frac{1}{Re} \cdot \frac{1}{\ln (R) +1}(Re)^{X-\epsilon}\\
& = \frac{1}{Re} \cdot \frac{1}{1+\ln (R)}(Re)^{X-\epsilon}\\
&=\frac{1}{\phi(X)}\cdot \frac{(Re)^{-\epsilon}}{1+\ln(R)}.
\end{align*}
It is easy to show by algebraic manipulation that the first term in the 3rd line
is equal to $0$.  Therefore,
\begin{align*}
&\frac{\OPT(\sigma)-2\cdot \ALG_\heterogen(\sigma)}{\ALG_\heterogen(\sigma)}\le\frac{1}{\phi(X)\sum_{(i,j)\in S}\Delta(x_i)/\phi(x_i)}\le (Re)^{\epsilon}(\ln(R)+1).
\end{align*}
Thus, $\OPT(\sigma) \le ((Re)^\epsilon (\ln(R)+3)
\cdot\ALG_\heterogen(\sigma)$, as desired.
\end{proof}


\subsection{A Lower Bound on the Competitive Ratio}
To wrap up this section, we show that if the large market assumption is the only
assumption on the sequence of workers, then no algorithm (even randomized) can
achieve a constant competitive ratio. We prove the result for the special case
that all the tasks are homogeneous.

The proof is very similar to the lower bound for the competitive ratio in
a variant of the online knapsack problem studied by \citet{ZhouCL08}. In this variant,
items arrive one by one online and we have a knapsack with some known capacity.
Each item has a weight and a utility parameter and it is assumed that the
utility to weight ratio for all the items are within a bounded range.
Furthermore, it is assumed that the weight of each item is small compared to the
capacity of the knapsack (similar to our large market assumption).  
The algorithm have to decide whether to pick an item or not upon
arrival and the goal is to maximize the utility of the picked items while
satisfying the knapsack capacity constraint.  Since our setting when the tasks
are homogeneous is a special case of the online
knapsack variant, the lower bound for that problem might not necessarily provide
us with a lower bound. However, we show that with a bit of care, we can achieve the same
lower bound using a similar construction.

The main idea of the lower bound is to construct hard sequences of worker arrivals.
A hard sequence can be described as follows.  The sequence starts with workers
with maximum bid, $R$, and then the following workers progressively have smaller
and smaller bids compared to the preceding workers.  Then at some random point
until the end of the sequence only workers with bid $R$ appear in the sequence.
Intuitively, these sequences are hard because no algorithm can foresee whether
the bids will decrease (so it should wait for cheaper workers) or increase (so it should spend
its budget on the current workers).

\iffull
\begin{restatable}{theorem}{thmlb}
\else
\begin{theorem}
\fi
\label{pro:lb1d}
For any (possibly randomized) online algorithm, there exists a set of sequences
of worker arrivals satisfying the large market assumption (all bids in $[1, R]$ and $R \ll B$) 
such that the competitive ratio of the algorithm on the sequence is at least
$\Omega\left(\ln(R)\right)$.
\iffull
\end{restatable}
\else
\end{theorem}
\fi

\newcommand{\lbproof}{
\begin{proof}
We modify the proof of Theorem 2.2 by \citet{ZhouCL08} to fit
our problem formulation. 

We use Yao's \emph{minimax principle}; by constructing a distribution over a set
of instances and showing that no \emph{deterministic} algorithm can achieve a expected
competitive ratio which is better than $\ln(R)+1$ on these instances, Yao's
principle implies that no \emph{randomized} algorithm can beat this
competitive ratio in the worst-case~\citep{Yao77}.

To construct the distribution, fix $\eta\in(0,1)$, let $k$ be the smallest
integer such that $(1-\eta)^k\le 1/R$, and define $k+1$ instances indexed by
$I_0$ to $I_k$ as follows.  $I_0$ contains $B/R$ identical workers all with bids
equal to $R$. For all $u>0$, $I_u$ is $I_{u-1}$ followed by $B/(R(1-\eta)^u)$
workers all with bids equal to $R(1-\eta)^u$. Since these instances have
different length we pad all the instances with enough workers with bid $R$ so
that all the instances have the same length.

We specify a distribution $D$ by $k+1$ values $p_0, \ldots, p_k$ where $p_u$ denotes the
probability of occurrence of instance $I_u$. Let
\begin{align*}
p_0 = p_1 = \ldots &= p_{k-1} := \frac{\eta}{(k+1)\eta+1}
\quad \text{ and } \quad
p_k := \frac{1+\eta}{(k+1)\eta+1}.
\end{align*}

On this distribution of inputs, any deterministic algorithm is fully specified
by the fraction of budget spent on hiring workers with bid $R(1-\eta)^u$; call
each fraction $f_u$.
Since the optimal assignment for instance $i$ is to only
hire workers with bids $R(1-\eta)^i$, the \emph{inverse} of the expected competitive ratio can be
bounded as follows.

\begin{align}
\label{eq:eqq}
\sum_{u=0}^k p_u \frac{\sum_{v=0}^u f_v B/(R(1-\eta)^v)}{B/(R(1-\eta)^u)}
&= \sum_{u=0}^k p_u \sum_{v=0}^uf_v(1-\eta)^{u-v}= \sum_{v=0}^k f_v \sum_{u=v}^k p_u (1-\eta)^{u-v},
\end{align}
where the last statement is by expanding the sums and reordering the terms.

The second sum in the RHS of~\Cref{eq:eqq} is bounded by 
\begin{align}
\label{eq:eqq2}
\sum_{u=v}^k p_u (1-\eta)^{u-v}
&= \frac{2\eta(1-\eta)^{k-v}+(1-\eta)}{(k+1)\eta+1}
\le \frac{2\eta+(1-\eta)}{(k+1)\eta+1}
= \frac{1+\eta}{(k+1)\eta+1},
\end{align}
where the first equality is derived exactly similar to~\citet{ZhouCL08}.
Replacing \Cref{eq:eqq2} into the RHS of~\Cref{eq:eqq},
\begin{align*}
\sum_{v=0}^k f_v \sum_{u=v}^k p_u (1-\eta)^{u-v}
&\le \frac{1+\eta}{(k+1)\eta+1}\sum_{v=0}^k f_v \le \frac{1+\eta}{(k+1)\eta+1}
\end{align*}
since by definition $\sum_{v=0}^k f_v\le 1$.

Now by definition of $k$, we know $(1-\eta)^k \le 1/R$, which implies $k+1\ge
\ln(R)/\ln(1/(1-\eta))$. So,
\begin{align*}
\sum_{v=0}^k &f_v \sum_{u=v}^k p_u (1-\eta)^{u-v}
\le \frac{1+\eta}{(k+1)\eta+1}
\le \frac{1+\eta}{\eta \ln(R)/\ln(1/(1-\eta)) + 1}
= O\left( \frac{1}{\ln(R)} \right)
\end{align*}
as $\eta \to 0$, because $\lim_{\eta\ra0} \eta/\ln(1/(1-\eta))=1$. Since we bound 
the inverse of the competitive ratio in the above analysis, then the
competitive ratio is at least $\Omega(\ln R)$.
\end{proof}
}

\iffull
Since the proof is modification of a proof by \citet{ZhouCL08}, we defer the
details to \Cref{sec:appendix}.
\else\lbproof\fi

Finally, it is worth mentioning that if we consider the limit where worker's bids are very small compared to the
budget ($\epsilon \to 0$), \Cref{thm:multiple-deadline-adv-1} shows that
\Cref{fig:4} has an optimal competitive ratio approaching $O(\ln(R))$, the best possible as proven by \Cref{pro:lb1d}.

\section{Random Permutation Setting}
\label{sec:perm}
Now that we have considered the worst-case scenario of inputs, let us consider
more ``well-distributed'' inputs.  We consider the \emph{random permutation
  model} \citep{DevanurH09}. In the random permutation model there is no
assumption about the bids of the workers (they can still be chosen by an
adversary) but we measure the competitive ratio a bit differently: we take all
possible permutations of these workers, and take the average competitive ratio
over all permutations. Intuitively, this assumption makes the task assignment easier because the
premium workers get distributed evenly in the sequence.

As pointed out by~\citet{DevanurH09}, the random permutation model can be
considered as drawing bids from an unknown distribution \emph{without replacement}.
Hence, the random permutation model is very similar to the model that assumes
the bids of the workers are drawn $\iid$ from an unknown distribution.


Throughout this section, we assume that the number of available workers, $n$, is
large and known to us. Also, we restrict our attention to inputs where the offline optimal algorithm
assigns at least a constant fraction of workers, i.e., $\OPT(\sigma) =
\Omega(n)$ for all sequences $\sigma$. 

For the case that the workers are homogeneous,~\citet{SingerM13} provide an
algorithm with competitive ratio of $360$. In this section, we extend their result
to the case the tasks are heterogeneous.  We also improve on their competitive
ratio, though the setting is a bit different: as they were concerned with
mechanism design properties, it was important for \citet{SingerM13} to carefully
tune their payments to (i) incentivize workers to bid honestly and (ii)
compensate all workers, even workers at the very beginning. 
Our algorithm only satisfies the first property. We discuss more about this at the end of this section.

If we knew the sequence of the workers and their bids, we could run $\OA$ to
compute a threshold price $p$ and number of assignments $Q$ where we know $Q$ is
at least a quarter of the optimal number of assignments by~\Cref{pro:42}.
However, since we are in the random permutation model, we can estimate this
threshold with high probability by observing a subset of workers. This idea is
summarized as the $\alg$ ($\RPA$) in~\Cref{fig:0}.
 
\begin{algorithm} 
\begin{algorithmic}
\STATE {{\bf Input}: Parameter $\alpha \in (0, 1)$, set of available tasks $J$, budget $B$.}
\STATE {{\bf Online input}:  Worker bids $\{b_{ij}\}$.}
\INDSTATE[1]{Let $C$ be the first half of workers; do not assign.}
\INDSTATE[1]{Let $\hat{p} := \OA(C,J,B/2)$.}
\INDSTATE[1]{On rest of input, run $\FTP((1 + \alpha)\hat{p},B/2)$.}
\end{algorithmic}
\caption{$\RPA$}
\label{fig:0}

\end{algorithm}

The algorithm observes the first half of workers, assign no tasks but computes a threshold $\hat{p}$ by running $\OA$ 
with a half of the budget
on the sequence of workers
on the first half. 
Given an input parameter $\alpha$, the algorithm then uses the threshold $(1+\alpha)\hat{p}$
and runs $\FTP$ on the second half of the workers with the remaining half of the
budget.\footnote{%
  While we can use the whole budget on the second half of the workers (as we did
  in our experiments) because the algorithm assigns no tasks to the first half
  of the workers, this will only decrease the competitive ratio by a factor of
  half.}
While using the threshold $(1+\alpha)\hat{p}$ instead of $\hat{p}$ might decrease
the performance of our algorithm, 
we show that $(1+\alpha)\hat{p}$ is higher than $p$ with high probability and use this observation
to make the analysis of the competitive ratio of $\RPA$ easier.
\begin{theorem}
\label{pro:4}
Let $\alpha, \delta \in (0, 1)$, and suppose the number of workers is at least
\[
  n = \Omega \left( \frac{1}{\alpha} \log(\frac{1}{\delta})\right) ,
\]
and $\OPT = \Omega(n)$ for every input.
For any sequence of workers $\sigma$, let $\ALG_\RPA(\sigma)$ and $\OPT(\sigma)$
denote the number of tasks assigned by the $\RPA$  and the offline optimal
algorithm for a sequence of workers $\sigma$, respectively. Then,
\[
  \OPT(\sigma)\le 8(1+\alpha)^2/(1-\alpha) \cdot\ALG_\RPA(\sigma),
\]
with probability at least $1-\delta$ for all $\sigma$.
\end{theorem}
\begin{proof}
If we knew the sequence of the workers and their bids, we could run $\OA$ to compute a price $p$ and number of assignments $Q$. 
We know $Q\ge \OPT/4$ by \Cref{pro:42}.  Let us refer to these $Q$ workers hired by $\OA$ as \emph{good workers}.

We first claim that if we use a price of $(1+\alpha)\cdot p$ instead of $p$, the number of assignments will decrease
by a factor of $(1+\alpha)$. This is because we can still assign a task to workers who $\OA$ assigned a task to with a threshold of $p$, but now we also have access 
to workers with bids in $(p, (1+\alpha)\cdot p]$ which may cause us to exhaust the budget faster.

Second, we can estimate $p$ from the first half of the workers.
Call this estimate $\hat{p}$. We claim that at least $(1-\alpha)Q/2$ of the good
workers are in the second half of the sequence with high probability, and
\begin{equation}
\label{eq-phat}
p \le (1+\alpha)\hat{p}\le (1+\alpha)p/(1-\alpha) .
\end{equation}
This together with the first claim and the 4-approximation of $\OA$ will give us
the competitive ratio claimed in the statement of the theorem.

To complete the proof, we first show that there are enough good workers in the second half of the sequence. Since we
are in the random permutation model, then with high probability, we know the
number of good workers in the first half of the workers is in $[Q(1-\alpha)/2,
Q(1+\alpha)/2]$. \citet{Chvatal79} show that this probability is bounded by
\begin{align*}
  \hspace{-5mm}
1-2\sum_{i=Q(1+\alpha)/2}^Q \frac{{i \choose Q/2}{n/2-i \choose n-Q}}{{n \choose n/2}}&=
\begin{cases}
1-O(e^{-\alpha Q^2/n}) & \text{ if } Q\ge \frac{n}{2},\\
1-O(e^{-\alpha (n-Q)^2Q/n^2}) & \text{ if } Q< \frac{n}{2}\\
\end{cases} = 1 - O(e^{-\alpha n}).
\end{align*}
To prove that \Cref{eq-phat} holds, we consider two cases.
\begin{enumerate}
\item The number of good workers in the first half is $[Q(1-\alpha)/2, Q/2)$.
  In this case $\hat{p}\ge p$ because we can obviously use the budget to assign
  tasks to all the good workers. And if there is leftover budget, we might be
  able to assign tasks to workers with bid greater than $p$. So,
\begin{align}
\hat{p}\frac{Q}{2}(1-\alpha) &\le \frac{B}{2} = \frac{pQ}{2}  \implies
\hat{p} \le \frac{1}{(1-\alpha)}p \notag,
\end{align}
since $B = PQ$ by the last line in \Cref{fig:9}. Combined with $\hat{p}\ge p$, the condition in~\Cref{eq-phat} holds in this case.

\item The number of good workers in the first half is $[Q/2, Q(1+\alpha)/2]$.

In this case $\hat{p}\le p$.  Suppose by contradiction that $\hat{p}>p$. This means the workers hired by the $\OA$
with bids of at most $\hat{p}$ is at least $S=B/(2\hat{p})$. We know $S\ge Q/2$ because all good workers have bids of 
most $p$. So,
\begin{align*}
\frac{B}{2\hat{p}} = \frac{pQ}{2\hat{p}} &\ge \frac{Q}{2} \implies
p \ge \hat{p},
\end{align*}
which is a contradiction. Thus,
\begin{align*}
  \hat{p} \frac{Q}{2}(1 + \alpha) &\geq \frac{B}{2}=\frac{pQ}{2} \implies
  \hat{p} \ge \frac{1}{1+\alpha}p.
\end{align*}
Combined with $\hat{p}\le p$, the condition in~\Cref{eq-phat} also holds in this
case, concluding the proof.
\end{enumerate}
\end{proof}

\paragraph*{A Note On Incentive Compatibility.}
A problem closely related to ours is to design \emph{incentive compatible
  mechanisms} in a setting where the workers arrive online but their bids are
private and unknown to the requester \citep{SingerM13}. While our focus in this
paper was mainly on designing competitive online algorithms, we point out that
both of our online algorithms, $\heterogen$ (\Cref{fig:4}) and $\RPA$
(\Cref{fig:0}), can lead to incentive compatible mechanisms in a straightforward
way. We briefly describe this relation here and omit the formal proofs. 

For any worker $i$, let $\phi_i$ denote the largest price that satisfies the potential $\phi$
upon the arrival of worker $i$.
Suppose in $\heterogen$, instead of checking the bids of each worker against
the potential function $\phi$, the requester simply offers $\phi_i$
to worker $i$ (without knowing the worker's true bids). One can simply
check that by allocating this budget to the $i$-th worker instead of her bid (which is never bigger than $\phi_i$ by definition)
in $\heterogen$, the proof of~\Cref{thm:multiple-deadline-adv-1} still
holds and consequently the modified $\heterogen$ still achieves 
the promised competitive ration. However, in this variant of
$\heterogen$, the dominant strategy for a worker is to reveal her true bid since
her utility (i.e, the payment she receives) is independent of her true bid ($0$
if $b_{ij} > \phi_i$ and $\phi_i$ otherwise). Similarly, we can argue this property
for the $\RPA$ (with the modification that we pay the workers equal to the fixed threshold price not their true bid) also. Although we are dealing with heterogeneous
tasks, the analysis of incentive compatibility is similar to~\citet{SingerM13} which only consider homogenous tasks.

\section{Experiments}

In this section we describe some experiments to examine the performance of 
$\heterogen$ (\Cref{fig:4}) and $\RPA$ (\Cref{fig:0}) on \emph{synthetic} 
data.\footnote{%
  We slightly modified the $\RPA$ in our experiments by using all the budget
  (instead of half) on the second half of the workers.  Note that this
  modification can only improve the performance in our implementation compared
  to our theoretical guarantee for \Cref{fig:0}.}
 
We conduct two sets of experiments. The first set focuses on the performance under \emph{homogeneous} and \emph{adversarially} chosen bids 
(similar to the hard distribution in \Cref{pro:lb1d}). In the second set of 
experiments, we use a sequence of heterogeneous workers and bids produced from a
family of distributions meant to model input in practical settings.

\paragraph{Adversarial Homogenous Bids.}
The input in this setting is created similar to the hard distribution in
\Cref{pro:lb1d} with groups of workers that are willing to solve \emph{any} of
the tasks with the same bid, but the bids may vary across different groups of workers.
Since these instances are hard inputs from the adversarial setting (as shown in
\Cref{pro:lb1d}), they are a good benchmark for comparing the algorithms in
adversarial setting.

\begin{figure}
\centering
\includegraphics[width=0.5\textwidth]{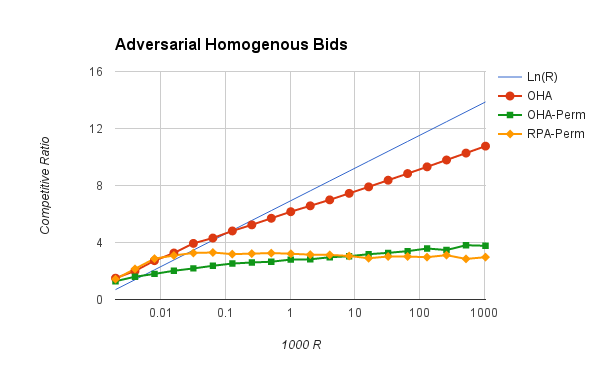}
\caption{Competitive Ratio of $\heterogen$ and $\RPA$ against an adversarially chosen input. The line with circle shows the competitive ratio of $\heterogen$ on the adversarial sequence. 
The lines with diamond and square show the competitive ratio of $\RPA$ and $\heterogen$ on the randomly permuted sequence, respectively.
A line for $\ln(R)$ as the theoretical bound on the competitive ratio of $\heterogen$ is also provided. The plot is presented in log-scale.
}
\label{exp:1}
\end{figure}

Formally, the input is created as follows. 
We vary $R$ in the range from $2$ to $2^{20}$ (on all powers of two) and set $B = 2R$. The bids of the workers are $R,R/2,\ldots,R/2^{i}$, 
where in each experiment, $i$ is chosen uniformly at random from $\{1, 2, \ldots, \log_2{R}\}$; for each group with bid $R/2^i$, there will be $B \cdot 2^i/R$ workers. 
We always pad the input with a set of workers with bid $R$ to ensure that the
total number of workers in each instance is the same (i.e., we set $n = 8R)$.
We set the number of tasks to be equal to the number of workers (so a worker will never run out of tasks to solve 
no matter what the previous workers have done). 
Finally, for each value of $R$, we create the input $10,000$ times and compute the average competitive ratio of $\heterogen$.
Note that since in this distribution, there is essentially no suitable worker in the second half, $\RPA$ performs very poorly and hence we do not 
include it in this part of the experiment. We then ran the same experiment (for both $\heterogen$ and $\RPA$) on the same input
but this time we permute the workers randomly; \Cref{exp:1} shows the result for this experiment. 

The experiment over the truly adversarial input shows that $\heterogen$ outperforms the theoretical guarantee of
$O(R^{\epsilon}\ln{(R)})$, ($\epsilon = 0.5$  here) and grows roughly like
$\ln{R}$, suggesting that the true performance of $\heterogen$ may be closer to
our theoretical lower bound from~\Cref{pro:lb1d}. 

The experiment over the randomly permuted input shows that $\RPA$ is obtaining a
constant competitive ratio even though the input does not satisfy all the 
properties required by \Cref{pro:4} (e.g., $\OPT = \Omega(n)$ assumption). Interestingly, even $\heterogen$ performs very 
well over the input (only slightly worse than $\RPA$), and also much better than the theoretical guarantee proved in~\Cref{thm:multiple-deadline-adv-1}. 
Finally note that although the sequence of workers were adversarially chosen to begin with, a random permutation of this sequence makes
the problem significantly easier; to the point that even $\heterogen$ which is designed to perform against worst case sequences will also 
achieve a constant competitive ratio.

\paragraph{Uniform Heterogeneous Bids.}
We consider a similar distribution (with minor modification of parameters) as the one introduced by~\citet{GoelNS14}. 
They use this setting to analyze the performance of their proposed algorithms for heterogeneous task assignment in the \emph{offline} case. These inputs are aimed to capture more 
realistic scenarios compared to our experiments in the previous part. 

\begin{figure}
\centering
\includegraphics[width=0.5\textwidth]{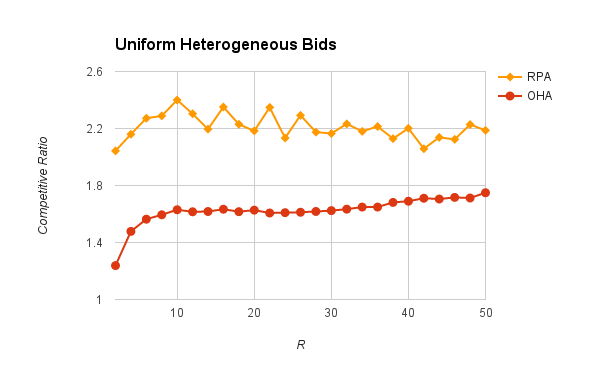}
\caption{Competitive Ratio of $\heterogen$ and $\RPA$ against the uniform heterogeneous bids. 
}
\label{exp:2}
\end{figure}

The distribution is as follows. We vary $R$ in range from $2$ to $50$ and set $B = 200$. The workers are heterogeneous and are only willing to solve a 
subset of all available tasks. In particular, we create $n = 200$ workers and $m = 200$ tasks and create a random graph with probability of edge formation equal to $0.05$ between 
any of the (worker, task) pairs. Additionally, over each realized edge, we choose the bid of the worker for the given task uniformly at random from  $\{1,2,\ldots, R\}$ 
(independently of the bids chosen over the other edges for this worker).
For each choice of $R$, we repeat the experiment $80$ times and report the average competitive ratio of the 
algorithms over the runs. \Cref{exp:2} illustrates the results of this experiment. 

As seen in \Cref{exp:2} both algorithms perform very well over these inputs,
obtaining a competitive ratio that is nearly independent of $R$. Interestingly,
$\heterogen$ performs distinctly better than its theoretical guarantee, even outperforming $\RPA$. 
We suspect that this is because $\heterogen$ is more adaptive than $\RPA$,
since it uses a
varying price threshold while $\RPA$ sticks to a fixed price (after
discarding half the input). Hence $\RPA$ 
exhausts the budget sooner than necessary, while additionally losing the 
contribution of half of the workers.

\paragraph{Acknowledgement}
We would like to thank Steven Zhiwei Wu for helpful discussion in the earlier stages of this work. 
We are also grateful to the anonymous reviewers for their insightful comments. 

\bibliography{bib}
\bibliographystyle{aaai}

\iffull
\appendix

\section{Proof of Lower Bound on Competitive Ratio}
\label{sec:appendix}

\thmlb*
\lbproof

\fi

\end{document}